\documentclass[11pt]{article}

\usepackage[margin=1in]{geometry}
\usepackage[T1]{fontenc}
\usepackage[utf8]{inputenc} % usually not needed on modern Overleaf, harmless
\usepackage{lmodern}
\usepackage{microtype}

\usepackage{amsmath, amssymb, amsthm, mathtools}
\usepackage{bm}        % bold math symbols
\usepackage{mathrsfs}  % script fonts
\usepackage{dsfont}    % indicator \mathds{1} if you want it
\usepackage{tikz}
\usetikzlibrary{arrows.meta,positioning}

\usepackage{enumitem}
\usepackage{booktabs}
\usepackage{array}
\usepackage{graphicx}
\usepackage{subcaption}
\usepackage{comment}

\usepackage{algorithm}
\usepackage{algpseudocode}

\usepackage[hidelinks]{hyperref}
\usepackage[nameinlink,capitalize,noabbrev]{cleveref}

\usepackage{xcolor}
\usepackage{todonotes}
\theoremstyle{plain}
\newtheorem{theorem}{Theorem}[section]
\newtheorem{lemma}[theorem]{Lemma}

\newtheorem{corollary}[theorem]{Corollary}

\theoremstyle{definition}
\newtheorem{definition}[theorem]{Definition}

\theoremstyle{remark}

\crefname{assumption}{Assumption}{Assumptions}

\title{A Tractable Class of Cooperative Games Defined by Directed Networks: Unanimity Decomposition and Shapley Value}
\author{ David Ryz\'{a}k\(^{\dagger}\)\and  Tomá\v{s} Kroupa\(^{\dagger}\)} 

\date{\(^{\dagger}\)Faculty of Electrical Engineering, Czech Technical University in Prague, Technická 2, Prague, 166 27, Czech Republic} % leave empty for conference style

\begin{document}
\maketitle

\begin{abstract}
We introduce a class of cooperative games induced by weighted directed graphs. Specifically, the coalitional value combines an internal interaction term given by the induced subgraph game with an external component based on minimal incoming edges from outside the coalition. The resulting game has a convenient representation in terms of unanimity games. This representation enables closed-form polynomial-time formulas for the Shapley and Banzhaf values. We further establish that the game has a nonempty core and is totally balanced. The class of such games therefore provides an analytically and computationally tractable example of structured  network-induced cooperative games in which stability-based allocations and fairness-based solution concepts do not coincide.
\end{abstract}


\begin{comment}

\subsubsection{What do we have so far}

\begin{itemize}

\item Mobius transform of the trust game. 
\item Shapley and Banzhaf polynomial formulas and their interpretation with respect to the newtwork.
\item Interpretation of shapley or banzhaf change with respect to single value $a_{ij}$ change.
\item  Core, Total balancedness, ordinal equivalence of Banzhaf and Shapley.
\item Brief monotonicity discussion.
\end{itemize}

\subsubsection{What do to with what we have so far}
\begin{itemize}
    \item Discuss general aggregation of  internal and external game such as 
    \begin{equation*}
        \alpha INT(S)+(1-\alpha)EXT(S), \ \alpha\in[0,1]
    \end{equation*}
    Solutions based on alpha: Results on Values and unanimity decomposition are simply generelizable for the lambda aggregation.
    
    \item vector valued trust, where current results are possibly the aggregation as above. Solutions in form: partial order of efficient solutions, lexiocographic order(at first sight it's not a good idea),  
    \item equilibria 
\end{itemize}

\subsubsection{Possible selling points}
\begin{itemize}
    \item Computational (closed form shapley, banzhaf)- We have 
     \item Aggregation mechanism using also the information from the outside of the coalition (Comment on why internal or external parts alone are not enough an why)
     \item Values (Banzhaf, Shapley) do not lose information about network but leverage it in specific way.
     On the other hand, core lose the network information but is stable.
\end{itemize}

\end{comment}
\section{Introduction}

We introduce a cooperative game induced by a weighted directed graph. The underlying network represents local evaluations between agents, which need not be symmetric. The coalitional worth combines two components: an internal term capturing interactions within coalition, and a non-additive external term determined by minimal incoming evaluations from outside the coalition.

The internal component alone is, in fact, an induced subgraph game whose Shapley value coincides with symmetric weighted degree centrality. In contrast, the external component acts as a bottleneck operator and, taken in isolation, neglects inner-coalitional interactions. Their sum produces a game that reflects both internal interaction and external exposure.

We analyze the resulting game through classical solution concepts. In particular, we derive explicit polynomial-time formulas for the Shapley and Banzhaf values, establish a unanimity decomposition of the game, and study properties including core characterization and total balancedness.

To justify the construction, it is important to observe that neither component alone yields a sufficiently rich class of games. The internal term defines an induced subgraph game whose Möbius transform contains only pairwise contributions. In particular, no higher-order coalition effects arise and the associated Shapley value reduces to symmetric weighted degree centrality.

%The external term, taken in isolation, defines a bottleneck game and its value depends only on minimal incoming edges from outside the coalition and therefore ignores effects among coalition members. In this case, coalition worth is determined entirely by boundary structure rather than internal interaction.
%This interaction produces nontrivial higher-order effects while preserving a tractable unanimity decomposition.

The external component takes the worth of a coalition based on edges from players outside the coalition, thus capturing exposure to external assessments rather than internal cooperation. Considered alone, this term assigns value independently of interactions among coalition members. Its combination with the internal interaction term generates nontrivial marginal effects while preserving an  unanimity decomposition.

A key structural property of the game is that its Möbius transform can be expressed in simple terms. The external component induces a chain of unanimity games whose supports form an inclusion chain determined by the ordering of incoming neighbors. This chain structure allows a tractable computation of classical solution concepts, e.g., Shapley and Banzhaf values.

Our main theoretical contributions are as follows.
First, we derive the Möbius transform of the game and obtain  polynomial-time formulas for the Shapley and Banzhaf values.
Second, we characterize the marginal impact of individual local evaluations on these value operators as a function of a player’s position in the network.
%\textcolor{red}{Third, we show that the Shapley and Banzhaf values are not ordinally equivalent within this class.}
Finally, we establish the structural properties of the game, including a characterization of the core and total balancedness.

In the following paragraph we briefly discuss related work. 
\begin{comment}
Then we define trust game and derive the Möbius transform for it. Further, we derive closed formula for the Shapley value of trust game and discuss the marginal contributions of weights on the Shapley value. Then we derive the same for Banzhaf value. At last we prove that trust game has singleton core and is totally balanced under the condition of weigths being nonegative.
\end{comment}
\paragraph{Related work}
Mathematically, our work is closely connected to studies on the complexity of cooperative games and games induced by a weighted graphs. The computational complexity of cooperative games is a crucial topic, since the cooperative game naturally has an exponential size in the number of players. Thus, a tractable computation of solution concepts is important. Chalkiadakis and Wooldridge in~\cite{chalkiadakis_weighted_2016} comment on Shapley value being $\sharp P$ in the context of weighted voting games, i.e., on a restricted class of games. The results of Deng and Papadmitriou in~\cite{deng_complexity_1994} shows that solution concepts such as the Shapley value can be computed in polynomial time in some cases, for instance, a weighted undirected graph. The work of Chalkiadakis, Elkind and Wooldridge~\cite{chalkiadakis_computational_2011} serves as a basic reference material on computational challenges in cooperative games, studying both computational complexity and the closely related topic of succinct game representation. Another succinct representations of cooperative games except of graph induced games are also MC-nets studied already by Ieong and Shoham~\cite{ieong_marginal_2005}. 
The graph based games are studied generally from several perspectives and we aim to comment on few of them. One of the first works of graphs studied via cooperative games is work~\cite{myerson_graphs_1977} with Myerson value as a flagship solution concept. 
Survey of Borm et al.~\cite{borm_operations_2001} studied several types of cooperative games induced by graphs in areas of operations research. 
Since our solution concepts are based on local graph relationships the topic is loosely connected to the degree centrality measures. The survey of Tarkowski et al.~\cite{tarkowski_game-theoretic_2017} summarize degree centrality approaches. More recently, Rusinowska and Van den Brink~\cite{brink_degree_2022,brink_degree_2024} study degree centrality from an axiomatic perspective in the context of graphs and digraphs.

The remainder of the paper is organized as follows. 
We briefly discuss related work in the following paragraph.
In Section~\ref{sec:model}, we introduce the trust game, establish its basic properties and derive the unanimity decomposition. In Sections~\ref{sec:values} and~\ref{sec:banz} we provide closed-form expressions for the Shapley and Banzhaf values and interpretation of their marginal effects in terms of the underlying weigthed graph and its local structure. 
Section~\ref{sec:core} investigates stability properties of the game, showing that the core is a singleton and establishing that the game is totally balanced. 
Finally, Section~\ref{sec:conclusion} concludes.

\section{Trust Game}\label{sec:model}

We start with the definition of the game which we also call trust game as in Bandhana et al.\cite{bandhana_trust_2024}
\begin{definition}[Trust Game]
Let $G = (N, E)$ be a weighted directed graph, where each edge $(i, j) \in E$ is assigned a weight $a_{ij} \in [0, 1]$.
Let $(N,v)$ be a \emph{trust game} given by 
$$v(S)=\underbrace{\sum_{\substack{i\in S,j\in S \\ (i,j)\in E}}a_{ij}}_{\text{internal}}+\underbrace{\sum_{i\in S^*}\min_{j\notin S:(j,i)\in E}a_{ji}}_{\text{external}},\ \text{for each}\ S\subseteq N $$ where $S^*=\{i\in S:$there exists $j\notin S$ such that $(j,i)\in E\}$.
\end{definition}
Bandhana et al.~\cite{bandhana_trust_2024} established some elementary properties of this game.
\begin{itemize}
    \item Trust game is \emph{superadditive}, i.e. for all disjoint coalitions $S,T \subseteq N$,
    \[
        v(S\cup T)\;\ge\; v(S)+v(T).
    \]
    \item Trust game is \emph{monotone}, i.e.  for all coalitions $S,T \subseteq N$ with $S\subseteq T$,
    \[
        v(S)\;\le\; v(T).
    \]
    \item There exists at least one allocation in the core of the trust game.
\end{itemize}
Next, we introduce the notation for the trust game, which we will use repeatedly in what follows.
For every player $i \in N$, we denote by
\begin{align*}
N^{-}(i) = \{\,i_1, i_2, \dots, i_{m_i}\,\}
\end{align*}
the set of players that have an incoming edges to $i$ (the \emph{in-neighbours} of $i$) and we have indexed them in nondecreasing order, i.e., $a_{i_1 i}\leq a_{i_2 i}\leq \ldots \leq a_{i_{m_i}i}$.

For the notational convenience and to stress the ranking among in-neigbours we use the following notation $b_i(t):=a_{i_t i}$, thus 
\begin{align*}
b_i(1) \le b_i(2) \le \dots \le b_i(m_i).
\end{align*}
We set $b_i(0):=0$ for convenience. %Therefore, the players in $N_i^-$ are ordered in such a way that $i_r$ has the $r$-th smallest incoming edge value to $i$.

\medskip
Let us now make more precise definition of the external game with respect to possibility of no in-neighbors for a given player $i$ in coalition $S$.
\begin{definition}[External game of the player $i$]
We denote by $w_i(S)$ the \emph{external game of player $i$} which is defined as follows:
\begin{align}
w_i(S):=\min_{j\notin S:(j,i)\in E}a_{ji} =
\begin{cases}
\displaystyle \min_{i_t \notin S} b_i(t), & \text{if } i \in S,\\[6pt]
0, & \text{if } i \notin S,
\end{cases}
\label{eq:def_vi}
\end{align}
where we used the convention that the minimum of the empty set of in-neighbors is $0$.  
\end{definition}
Next, we express $w_i$ as a linear combination of unanimity games. Recall that for any $T\subseteq N$, the unanimity game $u_T$ is defined by
\begin{align*}
u_T(S) =
\begin{cases}
1, & \text{if } T \subseteq S,\\
0, & \text{otherwise.}
\end{cases}
\end{align*}
Before moving onto the unanimity decomposition, we fix the notation for the coalitions of in-neighbors. For each $i\in N$, we denote by 
\begin{equation*}
    T_i(t) = \{\,i, i_1, i_2, \dots, i_{t-1}\,\}, \qquad t = 1,\dots, m_i+1,
\end{equation*}
the set of node $i$ and its first $t-1$ in-neighbors ordered by their nondecreasing weights.
The following lemma is about unanimity decomposition of the game using the Möbius transform~\cite{grabisch_aggregation_2009}.
\begin{lemma}[Unanimity decomposition of the external game of the player $i$]\label{lemma:ext_decomp}
Let
\begin{align*}
T_i(t) = \{\, i, i_1, i_2, \dots, i_{t-1} \,\}, 
\qquad t = 1,\dots, m_i+1,
\end{align*}
where $i_1, i_2, \dots, i_{t-1}$ denote the in-neighbors of player $i$, ordered in nondecreasing order. Let $u_T$ be a unanimity game on $T\subseteq N$.
Then for every coalition $S \subseteq N$,
\begin{align}
w_i(S)
= \sum_{t=1}^{m_i} [\big(b_i(t)-b_i(t-1)\big)\,u_{T_i(t)}(S)]
  - b_i(m_i)\,u_{T_i(m_i+1)}(S).
\label{eq:decomposition_external}
\end{align}
\end{lemma}

\begin{proof}
Consider a coalition $S\ni i$ and let
\begin{align*}
r = \min\{\,t : i_t \notin S\,\}.
\end{align*}

Then, from Equation~\eqref{eq:def_vi},
\[
w_i(S) =
\begin{cases}
b_i(r), & r \le m_i,\\
0, & N_i^- \subseteq S.
\end{cases}
\]
That is, $w_i(S)=0$ in case there is no in-neighbor of $i$ outside the coalition $S$. 
We connect the minimum operator to the unanimity games. For each $t \le m_i$ and a given $S\subseteq N$, the following equivalence holds \begin{equation*}
   w_i(S) = b_i(t)
\iff i,i_1,\dots,i_{t-1}\in S,\  \text{i.e.,}\ u_{T_i(k)}(S)=1 \ \forall k\in \{1,2,\ldots,t\}.  
\end{equation*}
Hence we can rewrite the minimum operator as a telescopic sum of unanimity games as follows:
\[
\sum_{t=1}^{m_i} \big(b_i(t)-b_i(t-1)\big) u_{T_i(t)}(S)
=
\begin{cases}
b_i(r), & r \le m_i,\\
b_i(m_i), & N_i^- \subseteq S.
\end{cases}
\]
However, $w_i(S)=0$ if $N_i^- \subseteq S$. By subtracting the term $b_i(m_i) u_{T_i(m_i+1)}(S)$ we obtain~\eqref{eq:decomposition_external}.

\begin{comment}
\[
\sum_{t=1}^{m_i} \big[(b_i(t)-b_i(t-1)\big) u_{T_i(t)}(S)]-b_i(m_i) u_{T_i(m_i+1)}(S)
=
\begin{cases}
b_i(r), & r \le m_i,\\
0, & N_i^- \subseteq S,
\end{cases}
\]
we have the unanimity game decomposition of $w_i(S)$.
\end{comment}

\end{proof}

\section{Shapley Value}\label{sec:values}
We first analyze how the unanimity decomposition of the external game $\sum_{i\in N}w_i$ determines the Shapley value of this part of the game, as its structure is less transparent than that of the internal component.
We denote by $r_j(i)$ the rank of player $i$ among the in-neighbors of $j$, that is,
\[
b_j\big(r_j(i)\big) = a_{ij}.
\]
Whenever no ambiguity arises, we write $r_j$ instead of $r_j(i)$, since only the rank of player $i$ appears explicitly in the formula for their Shapley value.

\begin{theorem}[Shapley value of the external game]
The Shapley value of the external game $\sum_{i\in N} w_i$ is as follows:
     \begin{align}\label{thm:shapley_external_thm}
 \phi_{i}\left(\sum_{j \in N}w_j\right)
= \sum_{t=1}^{m_i} \frac{b_i(t)-b_i(t-1)}{t}
  - \frac{b_i(m_i)}{m_i+1}+\sum_{j\neq i:(i,j)\in E}\ \sum_{t=r_j+1}^{m_i} \frac{b_j(t)-b_j(t-1)}{t}
  - \frac{b_j(m_j)}{m_j+1}.
\end{align}

\end{theorem}
\begin{proof}
The Shapley operator $\phi$ is linear.
Using Lemma~\ref{lemma:ext_decomp} yields
\begin{align}
\phi_i(w_k)
= \sum_{t=1}^{m_i} \big(b_i(t)-b_i(t-1)\big)\,\phi_i(u_{T_i(t)})
  - b_i(m_i)\,\phi_i(u_{T_i(m_i+1)}),\ k \in N.
\label{eq:shapley_in_thm}
\end{align}

For unanimity games, the Shapley value has the following form
\begin{align*}
\phi_i(u_T) =
\begin{cases}
\dfrac{1}{|T|}, & \text{if } i\in T,\\[4pt]
0, & \text{if } i \notin T.
\end{cases}
\end{align*}
Let us now discuss for which $k \in N$ is~\eqref{eq:shapley_in_thm} equal to $0$. If player $i$ is not among the in-neighbors of $k$ then all the terms in the sum of~\eqref{eq:shapley_in_thm} will be  $0$. Thus, we only need to calculate the sum for the players for which $(i,j)\in E$ and for the player $i$.

 Let us begin by computing $\phi_i(w_i)$. Since \(i \in T_i(t)\) for all \(t\), we have
\begin{align}
\phi_i(w_i)
= \sum_{t=1}^{m_i} \frac{b_i(t)-b_i(t-1)}{t}
  - \frac{b_i(m_i)}{m_i+1}.
\label{eq:phi_i}
\end{align}

When $i$ is an in-neighbour of $j$, the expression $\phi_i(w_j)$  is calculated as follows
\begin{align}
\phi_{i}(w_j)
= \sum_{t=r_j+1}^{m_i} \frac{b_j(t)-b_j(t-1)}{t}
  - \frac{b_j(m_j)}{m_j+1}.
\label{eq:phi_ir}
\end{align}
If $i$ is not in-neighbor of $j$ then $\phi_i(w_j)=0$.

Therefore, the Shapley value~\eqref{thm:shapley_external_thm} can be expressed using~\eqref{eq:phi_i} and~\eqref{eq:phi_ir}.
\begin{comment}
  \begin{align}
 \phi_{i}\left(\sum_{j \in N}w_j\right)
= \sum_{t=1}^{m_i} \frac{b_i(t)-b_i(t-1)}{t}
  - \frac{b_i(m_i)}{m_i+1}+\sum_{j\neq i:(i,j)\in E}\ \sum_{t=r_j+1}^{m_i} \frac{b_j(t)-b_j(t-1)}{t}
  - \frac{b_j(m_j)}{m_j+1}.
\label{eq:phi_whole}
\end{align}
\end{comment}

\end{proof}

\begin{comment}
\paragraph{Interpretation.}
Each increment \(\Delta_i(t) = b_i(t)-b_i(t-1)\)
represents the increase in the minimal external trust toward \(i\)
once all weaker--trusting neighbours
\(i_1,\dots,i_{t-1}\) have joined the coalition.
In the unanimity component \(u_{T_i(t)}\),
this increment \(\Delta_i(t)\) is shared equally among all players in \(T_i(t)\).
The final correction term
\(-b_i(m_i)u_{T_i(m_i+1)}\)
ensures that \(v_i\) becomes zero
once \(i\) and all its in--neighbours are in the coalition.
\end{comment}

\begin{lemma}[The unanimity decomposition of the trust game]
    
Let
\begin{align*}
T_i(t) = \{\, i, i_1, i_2, \dots, i_{t-1} \,\}, 
\qquad t = 1,\dots, m_i+1,
\end{align*}
where $i_1, i_2, \dots, i_{t-1}$ denote the in-neighbors of player $i$, ordered in nondecreasing order. Let $u_T$ be a unanimity game on $T\subseteq N$.
Then for every coalition $S \subseteq N$,
\begin{align}
v(S)
=\sum_{\substack{i\in S,j\in S \\ (i,j)\in S}}a_{ij}u_{\{ij\}}(S)+ \sum_{i\in S}\sum_{t=1}^{m_i} \big(\big(b_i(t)-b_i(t-1)\big)\,u_{T_i(t)}(S)\big)
  - b_i(m_i)\,u_{T_i(m_i+1)}(S).
\label{eq:decomposition}
\end{align}
\end{lemma}

\begin{proof}
The first two sums follow from rewriting the total weight of coalition $S$ as the sum of weights over all pairs of players within $S$. 
The second term follows from Lemma~\ref{lemma:ext_decomp}.
\end{proof}

\begin{corollary}[Shapley value of the trust game]
Let $(N,v)$ be  a trust game. The Shapley value of a player $i$ in a trust game $v$ is:
     \begin{align}
 \phi_{i}(v)
= \sum_{\substack{j\in N \\ (i,j)\in E,(j,i)\in E}}\dfrac{a_{ij}+a_{ji}}{2}+&\sum_{t=1}^{m_i} \frac{b_i(t)-b_i(t-1)}{t}
  - \frac{b_i(m_i)}{m_i+1}  \\
  +&\sum_{j\neq i:(i,j)\in E}\ \sum_{t=r_j+1}^{m_i} \frac{b_j(t)-b_j(t-1)}{t}
  - \frac{b_j(m_j)}{m_j+1}.
\end{align}
\end{corollary}

In the following lemma, we summarize marginal effects of weight $a_{ij}$ on the Shapley value and distinguish three cases according to the relative network positions of the players.
Let us begin by defining the marginal effect.

\begin{definition}[Marginal effect]
Let $(N,v)$ be a trust game and $G=(N,E)$ the corresponding weighted graph. Let $(i,j)\in E$ and $G_\varepsilon$ be the weighted graph identical to $G$ except
that the weight of edge $(i,j)$ is replaced by $a_{ij}+\varepsilon$ where $\varepsilon\in \mathbb{R}$ is small enough such that all rankings of in-neighbors are the same in $G$ and $G_\varepsilon$. Trust game corresponding to $G_\varepsilon$ is denoted by $(N,v_\varepsilon)$.
The \emph{marginal effect} of $a_{ij}$ on the Shapley value of
player $k$ is
\[
    \delta_{ij}^{k}(\varepsilon)
    \;:=\;
    \phi_k(v_\varepsilon)-\phi_k(v).
\]
\end{definition}

Thanks to linearity of the Shapley value we may discuss effects for internal and external game alone. The effect on the Shapley value of the trust game is just their sum.

\begin{lemma}[Effects of weight on Shapley value of player $i$]

 Marginal effects of the local value $a_{kj}, \ k,j \in N,\ k\neq j$ to the player's $i$ Shapley value:
\begin{itemize}
    \item If $(j,i)\in E$ then marginal effect of $a_{ji}$ is $\dfrac{1}{2}a_{ji}$ from internal game and $\dfrac{1}{r_i(j)(r_i(j)+1)}a_{ji}$ from the external game,
    \item If $(i,j)\in E$ then marginal effect of $a_{ij}$ is $\dfrac{1}{2}a_{ij}$ from internal game and $-\dfrac{1}{(r_j(i)+1)}a_{ij}$ from the external game,
    \item If $(i,j)\in E$ and $(k,j)\in E$ and $r_j(i)<r_j(k)$  then marginal effect of $a_{kj}$ is 0 from internal game and $\dfrac{1}{r_j(k)(r_j(k)+1)}a_{kj}$ from the external game.
\end{itemize}
\end{lemma}
\begin{proof}
 Compute directly from the Shapley value.
\end{proof}

\paragraph{Interpretation of Shapley value}
Viewing the game through marginal effects makes the interpretation of the Shapley value in the trust game easier to understand. In particular, the Shapley value of an agent $i$ decomposes into three classes of marginal effects:
\begin{enumerate}
    \item \textbf{Self-contribution:} the effect of $i$'s own participation on $i$'s value.
    \item \textbf{Incoming-neighbor effects:} effects associated with agents $j$ such that $(j,i)\in E$.
    \item \textbf{"Shared-neighbor" effects:} effects associated with agents $k$ for which there exists some $j$ with $(i,j)\in E$ and $(k,j)\in E$.
\end{enumerate}
Equivalently, only trust information that is local to player $i$, edges incident to $i$ (in either direction) and edges from other agents to the out-neighbors of $i$ can affect the Shapley value of the player $i$.

We consider a change in a single weight $a_{ij}$ while keeping all other weights fixed. We then distinguish three cases according to the network position of the player whose Shapley value is affected.

The graph in Figure~\ref{fig:discrete_graph} illustrates this particular setting. Players whose Shapley values change are highlighted as colored vertices.
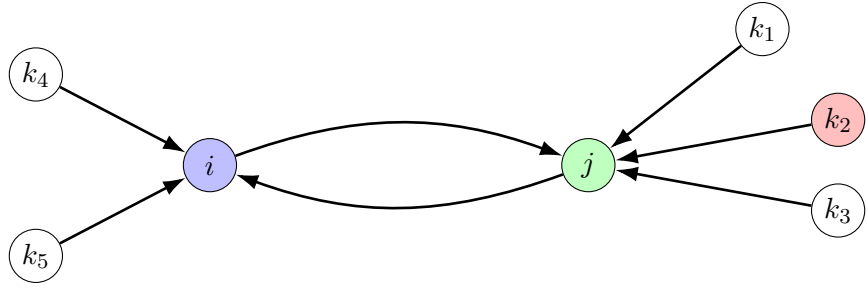
\begin{figure}[ht]
\centering
\begin{tikzpicture}[
  v/.style={circle,draw,minimum size=20pt,inner sep=1pt},
  e/.style={-{Latex[length=3mm,width=2mm]}, line width=1pt}
]

% --- vertices ---
\node[v, fill=blue!25]  (i)  at (-2.5,0) {$i$};
\node[v, fill=green!25] (j)  at ( 2.5,0) {$j$};

\node[v]               (k1) at (4.8, 1.8) {$k_1$};
\node[v, fill=red!25]  (k2) at (5.8, 0.6) {$k_2$};
\node[v]               (k3) at (5.8,-0.6) {$k_3$};

\node[v] (k4) at (-4.8, 1.2) {$k_4$};
\node[v] (k5) at (-4.8,-1.2) {$k_5$};

% --- directed edges: ij, ji, k1j, k2j, k3j, k4i, k5i ---

% ij and ji (separate by bending so both arrows are visible)
\draw[e, bend left=20] (i) to (j); % ij
\draw[e, bend left=20] (j) to (i); % ji

% k_1j, k_2j, k_3j
\draw[e] (k1) to (j);  % k1j
\draw[e] (k2) to (j);  % k2j
\draw[e] (k3) to (j);  % k3j

% k_4i, k_5i
\draw[e] (k4) to (i);  % k4i
\draw[e] (k5) to (i);  % k5i

\end{tikzpicture}
\caption{Graph example for marginal effects.}
\label{fig:discrete_graph}
\end{figure}
In Figures~\ref{fig:marg_graph_i}, \ref{fig:shapley_j}, and \ref{fig:k_2shapley}, we omit explicit numerical values on the vertical axis, as they do not add greater insight in this illustrative setting. Instead, we display the relevant edge weights on the horizontal axis to highlight the dependence of the Shapley value on their change.
\begin{enumerate}
    \item In Figure~\ref{fig:marg_graph_i} the effect of the player's $i$ weight $a_{ij}$ to his Shapley value is shown. 
    \begin{figure}[ht]
    \centering
\begin{tikzpicture}[
  >=Latex,
  axis/.style={->, line width=0.9pt},
  bp/.style={dash pattern=on 3pt off 3pt, line width=0.9pt},
  hp/.style={dash pattern=on 3pt off 3pt, line width=0.8pt},
  tick/.style={line width=0.8pt},
  curve/.style={line width=1.2pt, blue!70},
  lab/.style={font=\small}
]

\def\Xscale{6}
\def\S{6}  % keep spacing
\def\Ytop{4.8} % ≈ 2/3 of previous 7.2

\def\xa{0.2}
\def\xb{0.5}
\def\xc{0.8}

\pgfmathsetmacro{\Xa}{\Xscale*\xa}
\pgfmathsetmacro{\Xb}{\Xscale*\xb}
\pgfmathsetmacro{\Xc}{\Xscale*\xc}
\pgfmathsetmacro{\Xone}{\Xscale*1.0}

% symbolic heights
\pgfmathsetmacro{\yA}{0}
\pgfmathsetmacro{\yBminus}{\S*(1/6*\xb)}
\pgfmathsetmacro{\yBplus}{\S*(1/4*\xb + 1/12)}
\pgfmathsetmacro{\yCminus}{\S*(1/4*\xc + 1/12)}
\pgfmathsetmacro{\yCplus}{\S*(3/10*\xc + (4/5)*(1/4))}
\pgfmathsetmacro{\yOne}{\S*(3/10*1 + (4/5)*(1/4))}

\begin{scope}[shift={(2.2,0)}]

% Shorter axes
\draw[axis] (0,0) -- (6.2,0) node[lab, right] {$a_{ij}$};
\draw[axis] (0,0) -- (0,\Ytop) node[lab, above, text=blue!75] {$\phi_i$};

% x ticks
\draw[tick] (0,0) -- (0,-0.10);
\node[lab, below] at (0,-0.12) {$0$};

\draw[tick] (\Xone,0) -- (\Xone,-0.10);
\node[lab, below] at (\Xone,-0.12) {$1$};

\draw[tick] (\Xa,0) -- (\Xa,-0.10);
\node[lab, below=6pt] at (\Xa,-0.12) {$a_{k_1 j}=0.2$};

\draw[tick] (\Xb,0) -- (\Xb,-0.10);
\node[lab, below=6pt] at (\Xb,-0.12) {$a_{k_2 j}=0.5$};

\draw[tick] (\Xc,0) -- (\Xc,-0.10);
\node[lab, below=6pt] at (\Xc,-0.12) {$a_{k_3 j}=0.8$};

% shorter dashed vertical guides
\draw[bp] (\Xa,0) -- (\Xa,\Ytop);
\draw[bp] (\Xb,0) -- (\Xb,\Ytop);
\draw[bp] (\Xc,0) -- (\Xc,\Ytop);

% y levels
\foreach \y/\label in {
    \yA/P_0,
    \yBminus/P_1,
    \yBplus/P_2,
    \yCminus/P_3,
    \yCplus/P_4,
    \yOne/P_5
}{
    \draw[tick] (0,\y) -- (-0.12,\y);
    \node[lab, left] at (-0.15,\y) {$\label$};
    \draw[hp] (0,\y) -- (\Xone,\y);
}

% curve
\draw[curve] (0,\yA) -- (\Xa,\yA);
\draw[curve] (\Xa,\yA) -- (\Xb,\yBminus);
\draw[curve] (\Xb,\yBplus) -- (\Xc,\yCminus);
\draw[curve] (\Xc,\yCplus) -- (\Xone,\yOne);

\end{scope}
\end{tikzpicture}
  \caption{Marginal effects on the Shapley value for player $i$ with other incoming links.}
    \label{fig:marg_graph_i}
\end{figure}
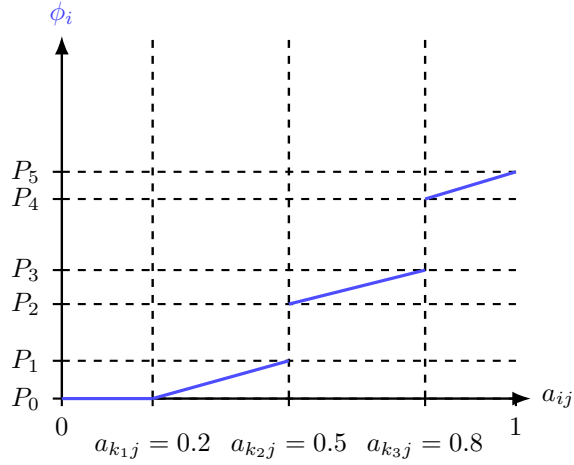
 \item In Figure~\ref{fig:shapley_j} the effect of the player's $i$ weight $a_{ij}$ to Shapley value of his out-neighbor $j$ is shown.

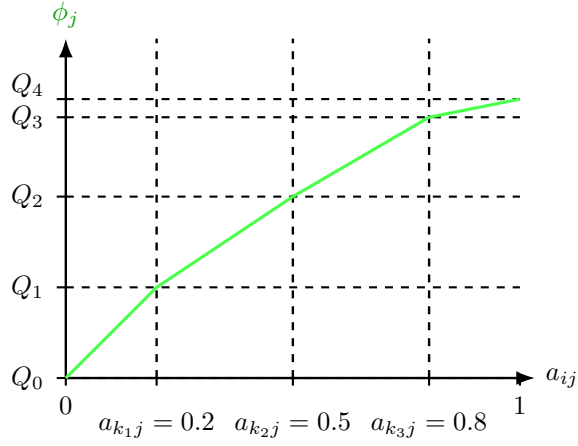
\begin{figure}[ht]
    \centering
\begin{tikzpicture}[
  >=Latex,
  axis/.style={->, line width=0.9pt},
  bp/.style={dash pattern=on 3pt off 3pt, line width=0.9pt},
  hp/.style={dash pattern=on 3pt off 3pt, line width=0.8pt},
  tick/.style={line width=0.8pt},
  curve/.style={line width=1.2pt, green!70},
  lab/.style={font=\small}
]

\def\Xscale{6}
\def\S{6}  % vertical scaling

% Continuous increments over Q0:
% inc(0.2)=1/5
% inc(0.5)=2/5
% inc(0.8)=23/40
% inc(1)=123/200

\pgfmathsetmacro{\yQzero}{0}
\pgfmathsetmacro{\yQone}{\S*(1/5)}
\pgfmathsetmacro{\yQtwo}{\S*(2/5)}
\pgfmathsetmacro{\yQthree}{\S*(23/40)}
\pgfmathsetmacro{\yQfour}{\S*(123/200)}

\pgfmathsetmacro{\Ytop}{\yQfour + 0.8}

% breakpoints
\def\xa{0.2}
\def\xb{0.5}
\def\xc{0.8}

\pgfmathsetmacro{\Xa}{\Xscale*\xa}
\pgfmathsetmacro{\Xb}{\Xscale*\xb}
\pgfmathsetmacro{\Xc}{\Xscale*\xc}
\pgfmathsetmacro{\Xone}{\Xscale*1.0}

\begin{scope}[shift={(2.2,0)}]

% Axes
\draw[axis] (0,0) -- (6.2,0) node[lab, right] {$a_{ij}$};
\draw[axis] (0,0) -- (0,\Ytop) node[lab, above, text=green!60!black] {$\phi_j$};

% x ticks
\draw[tick] (0,0) -- (0,-0.10);
\node[lab, below] at (0,-0.12) {$0$};

\draw[tick] (\Xone,0) -- (\Xone,-0.10);
\node[lab, below] at (\Xone,-0.12) {$1$};

\draw[tick] (\Xa,0) -- (\Xa,-0.10);
\node[lab, below=6pt] at (\Xa,-0.12) {$a_{k_1 j}=0.2$};

\draw[tick] (\Xb,0) -- (\Xb,-0.10);
\node[lab, below=6pt] at (\Xb,-0.12) {$a_{k_2 j}=0.5$};

\draw[tick] (\Xc,0) -- (\Xc,-0.10);
\node[lab, below=6pt] at (\Xc,-0.12) {$a_{k_3 j}=0.8$};

% dashed vertical guides
\draw[bp] (\Xa,0) -- (\Xa,\Ytop);
\draw[bp] (\Xb,0) -- (\Xb,\Ytop);
\draw[bp] (\Xc,0) -- (\Xc,\Ytop);

% ---- Y levels ----

% Q0
\draw[tick] (0,\yQzero) -- (-0.12,\yQzero);
\node[lab, left] at (-0.15,\yQzero) {$Q_0$};
\draw[hp] (0,\yQzero) -- (\Xone,\yQzero);

% Q1
\draw[tick] (0,\yQone) -- (-0.12,\yQone);
\node[lab, left] at (-0.15,\yQone) {$Q_1$};
\draw[hp] (0,\yQone) -- (\Xone,\yQone);

% Q2
\draw[tick] (0,\yQtwo) -- (-0.12,\yQtwo);
\node[lab, left] at (-0.15,\yQtwo) {$Q_2$};
\draw[hp] (0,\yQtwo) -- (\Xone,\yQtwo);

% Q3
\draw[tick] (0,\yQthree) -- (-0.12,\yQthree);
\node[lab, left] at (-0.15,\yQthree) {$Q_3$};
\draw[hp] (0,\yQthree) -- (\Xone,\yQthree);

% Q4 (label slightly lifted to avoid overlap)
\draw[tick] (0,\yQfour) -- (-0.12,\yQfour);
\node[lab, left, yshift=5pt] at (-0.15,\yQfour) {$Q_4$};
\draw[hp] (0,\yQfour) -- (\Xone,\yQfour);

% Continuous curve
\draw[curve]
  (0,\yQzero)
  -- (\Xa,\yQone)
  -- (\Xb,\yQtwo)
  -- (\Xc,\yQthree)
  -- (\Xone,\yQfour);

\end{scope}
\end{tikzpicture}
    \caption{Marginal effects on the Shapley value for player $j$ with other incoming links.}
    \label{fig:shapley_j}
\end{figure}
 
\item In Figure~\ref{fig:k_2shapley} the effect of the player's $i$ weight $a_{ij}$ to Shapley value of player $k_2$ who is also in-neighbor of player $j$ is shown.

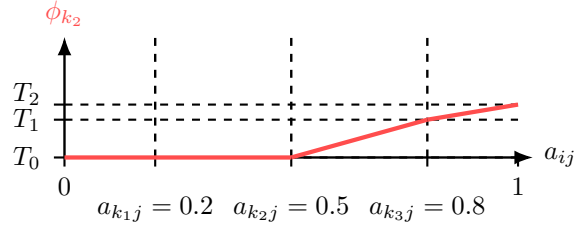
\begin{figure}[ht]
    \centering
\begin{tikzpicture}[
  >=Latex,
  axis/.style={->, line width=0.9pt},
  bp/.style={dash pattern=on 3pt off 3pt, line width=0.9pt},
  hp/.style={dash pattern=on 3pt off 3pt, line width=0.8pt},
  tick/.style={line width=0.8pt},
  curve/.style={line width=1.6pt, red!70},
  lab/.style={font=\small}
]

\def\Xscale{6}
\def\S{20}

% breakpoints
\def\xa{0.2}
\def\xb{0.5}
\def\xc{0.8}

\pgfmathsetmacro{\Xa}{\Xscale*\xa}
\pgfmathsetmacro{\Xb}{\Xscale*\xb}
\pgfmathsetmacro{\Xc}{\Xscale*\xc}
\pgfmathsetmacro{\Xone}{\Xscale*1.0}

% Heights
\pgfmathsetmacro{\yTzero}{0}
\pgfmathsetmacro{\yTone}{\S*((1/12)*(\xc-\xb))}                 
\pgfmathsetmacro{\yTtwo}{\S*((1/12)*(\xc-\xb) + (1/20)*(1-\xc))}

\pgfmathsetmacro{\Ytop}{\yTtwo + 0.9}

\begin{scope}[shift={(2.6,0)}]

% Axes
\draw[axis] (0,0) -- (6.2,0) node[lab, right] {$a_{ij}$};
\draw[axis] (0,0) -- (0,\Ytop) node[lab, above, text=red!70] {$\phi_{k_2}$};

% x ticks
\draw[tick] (0,0) -- (0,-0.12);
\node[lab, below] at (0,-0.14) {$0$};

\draw[tick] (\Xone,0) -- (\Xone,-0.12);
\node[lab, below] at (\Xone,-0.14) {$1$};

\draw[tick] (\Xa,0) -- (\Xa,-0.12);
\node[lab, below=7pt] at (\Xa,-0.14) {$a_{k_1 j}=0.2$};

\draw[tick] (\Xb,0) -- (\Xb,-0.12);
\node[lab, below=7pt] at (\Xb,-0.14) {$a_{k_2 j}=0.5$};

\draw[tick] (\Xc,0) -- (\Xc,-0.12);
\node[lab, below=7pt] at (\Xc,-0.14) {$a_{k_3 j}=0.8$};

% dashed vertical guides
\draw[bp] (\Xa,0) -- (\Xa,\Ytop);
\draw[bp] (\Xb,0) -- (\Xb,\Ytop);
\draw[bp] (\Xc,0) -- (\Xc,\Ytop);

% y levels

\draw[tick] (0,\yTzero) -- (-0.14,\yTzero);
\node[lab, left] at (-0.18,\yTzero) {$T_0$};
\draw[hp] (0,\yTzero) -- (\Xone,\yTzero);

\draw[tick] (0,\yTone) -- (-0.14,\yTone);
\node[lab, left] at (-0.18,\yTone) {$T_1$};
\draw[hp] (0,\yTone) -- (\Xone,\yTone);

% T2 (label shifted slightly upward only)
\draw[tick] (0,\yTtwo) -- (-0.14,\yTtwo);
\node[lab, left, yshift=4pt] at (-0.18,\yTtwo) {$T_2$};
\draw[hp] (0,\yTtwo) -- (\Xone,\yTtwo);

% Curve
\draw[curve] (0,\yTzero) -- (\Xb,\yTzero);
\draw[curve] (\Xb,\yTzero) -- (\Xc,\yTone);
\draw[curve] (\Xc,\yTone) -- (\Xone,\yTtwo);

\end{scope}
\end{tikzpicture}
    \caption{Marginal effects on the Shapley value for player $k_2$ with other incoming links.}
    \label{fig:k_2shapley}
\end{figure}
\end{enumerate}

From Figure~\ref{fig:marg_graph_i} we may conjecture the following, the player $i$ cannot obtain a greater Shapley value $\phi_i$ by changing weight  $a_{ij}$ if there are no other in-neighbors of player $j$. We now present formally this observation as a characterization of null players in a trust game.
\begin{lemma}[Null player in trust game]
    Let $(N,v)$ be a trust game. Player $i$ is a null player ($\phi_i(v)=0$) if and only if \begin{equation}
        \forall j \in N^-_i:a_{ji}=0 \And \forall k:i\in N_k^-: (|N^-_k|=1) \vee (\forall m\in N_k^-:a_{mk}=0) 
        \end{equation}
\end{lemma}

\begin{comment}

\begin{lemma}
    TODO:WHICH kind of monotonicity works. Does not work somethinth only concerning the links incoming to two different players.
\end{lemma}

\begin{lemma}
Graphs of change for the shapley value     
\end{lemma}
TODO :What are the incentives for the players in such a case....
\begin{lemma}
    dddd
\end{lemma}
\end{comment}

\begin{comment}
\subsection*{Almost Impartial shapley }
waht if we want the shapley to be almost impartial, such that the player may slightly change his local values and thus change his shapley value but if has no or limited information about local values of others he may not be able to do it at all. 
Idea 

\begin{defn}[One player ranking impartiality]
    Let $\varphi_i\leq \varphi_j$. Suppose $(i,j)\notin E$. By adding $a_{ij}$ the relative order of $\varphi_i$ and $\varphi_j$ cannot be changed.
\end{defn}
NOTE : how it is in dynamics?
\end{comment}

\section{Banzhaf value}\label{sec:banz}
We compute the Banzhaf value of the trust game. We use the Möbius transform of the trust game to derive the result. 
Let us calculate the Banzhaf value for the coalitions in unanimity game decompostion:
\begin{lemma}[Banzhaf value of a unanimity game]\label{lemma:banzhaf_unanimity}
Let $N$ be the player set with $|N|=n$, and let $u_T$ be the unanimity game
\[
u_T(S)=
\begin{cases}
1, & \text{if } T\subseteq S,\\
0, & \text{otherwise,}
\end{cases}
\qquad S\subseteq N.
\]
Then for every player $k\in N$, the  Banzhaf value satisfies
\[
\beta_k(u_T)=
\begin{cases}
2^{\,1-|T|}, & \text{if } k\in T,\\
0, & \text{if } k\notin T.
\end{cases}
\]
\end{lemma}

\begin{proof}
Fix $k\in N$. By definition,
\[
\beta_k(u_T)=\frac{1}{2^{n-1}}\sum_{S\subseteq N\setminus\{k\}}
\bigl(u_T(S\cup\{k\})-u_T(S)\bigr).
\]
We analyze when the increment $u_T(S\cup\{k\})-u_T(S)$ equals $1$.

\emph{Case 1: $k\notin T$.}
 For every $S\subseteq N\setminus\{k\}$,

$u_T(S\cup\{k\})=u_T(S)$,
 thus $\beta_k(u_T)=0$.
\\

\emph{Case 2: $k\in T$.}

 For $k\in T$,
\[
u_T(S\cup\{k\})-u_T(S)=1
\quad\text{if and only if}\quad
T\setminus\{k\}\subseteq S.
\]
So the sum counts the number of coalitions $S\subseteq N\setminus\{k\}$ that contain $T\setminus\{k\}$.

Let us count such coalitions $S$ for a given coalition $T$. Once $T\setminus\{k\}$ is included, the remaining elements of $S$ can be chosen arbitrarily from
\[
(N\setminus\{k\})\setminus (T\setminus\{k\}) = N\setminus T,
\]
which has size $n-|T|$. Hence the number of admissible $S$ is $2^{n-|T|}$.

Therefore,
\[
\beta_k(u_T)
= \frac{1}{2^{n-1}}\cdot 2^{n-|T|}
=2^{1-|T|}.
\]
\end{proof}

\begin{theorem}
    The Banzhaf value of player $i$ in trust game $(N,v)$ is:
    \begin{align*}
 \beta_{i}(v)
=&\sum_{\substack{j\in N \\ (i,j)\in E,(j,i)\in E}}\dfrac{a_{ij}+a_{ji}}{2}   \\
+& \sum_{t=1}^{m_i} \frac{b_i(t)-b_i(t-1)}{2^{t-1}}
  - \frac{b_i(m_i)}{2^{m_i}}+\sum_{j\neq i:(i,j)\in E}\ \sum_{t=r_j+1}^{m_i} \frac{b_j(t)-b_j(t-1)}{2^{t-1}}
  - \frac{b_j(m_j)}{2^{m_j}}.
\end{align*}
\end{theorem}
\begin{proof}
   We can follow the proof in the same manner as for the Shapley value. The only change in the proof is the Banzhaf value of the unanimity game which is calculated in Lemma~\ref{lemma:banzhaf_unanimity}
$$ \beta_i(u_T(S))=\frac{1}{2^{|T|-1}}.$$
   We can now simplify the formula such that 
      \begin{align}
 \beta_{i}(v)
=&\sum_{\substack{j\in N \\ (i,j)\in E,(j,i)\in E}}\dfrac{a_{ij}+a_{ji}}{2}   \\
+& \sum_{t=1}^{m_i} \frac{b_i(t)}{2^{t}}+\sum_{j\neq i:(i,j)\in E}\ \sum_{t=r_j+1}^{m_i} \frac{b_j(t)}{2^{t}}
  - \frac{a_{ij}}{2^{r_j}}.
\end{align}
\end{proof}

We present marginal effects of weights to the Banzhaf values analogously to the Shapley value.

\begin{lemma}

Marginal effects of the local value $a_{kj}$ to player's $i$ Banzhaf value:
\begin{itemize}
    \item If $(j,i)\in E$ then marginal effect of $a_{ji}$ is $\dfrac{1}{2}a_{ji}$ from internal game and $\dfrac{1}{2^{r_i(j)}}a_{ji}$ from the external game,
    \item If $(i,j)\in E$ then marginal effect of $a_{ij}$ is $\dfrac{1}{2}a_{ij}$ from internal game and $-\dfrac{1}{2^{r_j(i)}}a_{ij}$ from the external game,
    \item If $(i,j)\in E$ and $(k,j)\in E$ and $r_j(i)<r_j(k)$  then marginal effect of $a_{kj}$ is 0 from internal game and $\dfrac{1}{2^{r_j(k)}}a_{kj}$ from the external game.
\end{itemize}
\end{lemma}
\begin{proof}
    We can calculate the marginal effects by looking at the coefficient at a given weight $a_{kj}$ in the player's $i$ formula for the Banzhaf value. 
\end{proof}

\begin{comment}
\begin{lemma}
    Let $(N,v)$ be a trust game. Banzhaf and Shapley value of the game $(N,v)$ are not ordinally equivalent. 
\end{lemma}
\begin{proof}
   TODO: Counterexample or remove lemma and comments on the ordinal equivalency
\end{proof}
\end{comment}

\section{Core of The Trust Game}\label{sec:core}
We now turn to the properties of the trust game related to its stability under coalitional deviations. We first show that the core of the trust game is always a singleton. Moreover, this allocation does not, in general, coincide with either the Shapley value or the Banzhaf value.
\begin{lemma}[Core of trust game]\label{lemma:core}
Let $(N,v)$ be a trust game. The core of $(N,v)$ consists of the single allocation
\[
\left(\sum_{\substack{j\in N \\ (j,1)\in E}} a_{j1}, \sum_{\substack{j\in N \\ (j,2)\in E}} a_{j2}, \ldots, \sum_{\substack{j\in N \\ (j,n)\in E}} a_{jn}\right).
\]
\end{lemma}

\begin{proof}
The fact that the allocation 
\[
\left(\sum_{\substack{j\in N \\ (j,1)\in E}} a_{j1}, \sum_{\substack{j\in N \\ (j,2)\in E}} a_{j2}, \ldots, \sum_{\substack{j\in N \\ (j,n)\in E}} a_{jn}\right).
\]
belongs to the core follows directly from the structure of the trust game and is also established by Bandhana et al. in~\cite{bandhana_trust_2024}. 
We first establish the following identity:
\begin{equation}\label{eq:core}
    \sum_{i\in N} \frac{1}{n-1} v(N\setminus \{i\}) = v(N).
\end{equation}
    We can rewrite from the definition of the game \begin{equation*}
        v(N\setminus i)=\sum_{k \in N\setminus i}\sum_{\substack{l \in N\setminus i \\ (k,l)\in E}} a_{kl}+\sum_{\substack{k \in N\setminus i\\ (i,k)\in E}}a_{ik }.
    \end{equation*} From this we can see that \begin{equation*}
        v(N\setminus i)=\sum_{k \in N}\sum_{\substack{l \in N\\ (k,l)\in E}} a_{kl}-\sum_{\substack{k \in N\setminus i\\ (k,i)\in E}}a_{ki }.
    \end{equation*}
For each $i \in N$, the value $v(N \setminus \{i\})$ equals the sum of all edge weights in the graph except for those corresponding to edges entering $i$. Hence, in \eqref{eq:core}, each edge weight is counted exactly $n-1$ times on the left-hand side. Multiplication by the factor $\frac{1}{n-1}$ therefore yields the total sum of all edge weights in the graph, which coincides with $v(N)$.
%We now show that the unique core allocation is given by 
%\[
%x_i = v(N) - v(N \setminus \{i\}).
%\]
Let $x \in C(v)$ and define
\[
c_i := v(N) - v(N \setminus \{i\}).
\]
By coalitional rationality applied to the coalition $N \setminus \{i\}$, we have
\[
x(N \setminus \{i\}) \ge v(N \setminus \{i\}),
\]
which, together with efficiency $x(N)=v(N)$, implies
\[
x_i = v(N) - x(N \setminus \{i\}) \le v(N) - v(N \setminus \{i\}) = c_i.
\]
Hence $x_i \le c_i$ for all $i \in N$.
Summing over all players and using \eqref{eq:core}, we obtain
\begin{equation}
\sum_{i \in N} c_i
= \sum_{i \in N} \big( v(N) - v(N \setminus \{i\}) \big)
= n\, v(N) - \sum_{i \in N} v(N \setminus \{i\})
= v(N).
\end{equation}
Since $x \in C(v)$, efficiency implies
\[
\sum_{i \in N} x_i = v(N) = \sum_{i \in N} c_i.
\]
Consequently,
\[
\sum_{i \in N} (c_i - x_i) = 0.
\]
As $c_i - x_i \ge 0$ for all $i \in N$, the above equality can hold only if
\[
c_i = x_i \quad \text{for all } i \in N.
\]
Therefore, the core consists of a single allocation.
\end{proof}

Next, we discuss total balancedness of the trust game. 
\begin{lemma}[Total balancedness]
 Let $(N,v)$ be a trust game, where $\forall \ i,j\in N: a_{ij}\geq 0$. The trust game is totally balanced.
\end{lemma}
\begin{proof}
We show that for every coalition $S \subseteq N$, the restricted game $(S,v_S)$, defined by 
\[
v_S(T) = v(T), \qquad T \subseteq S,
\]
has a nonempty core. By the Bondareva–Shapley theorem idependently proved by Shapley~\cite{shapley_balanced_1967} and Bondareva~\cite{bondareva_applications_1963}, it is sufficient to establish total balancedness of the trust game.

We proceed by constructing an explicit core allocation for each $(S,v_S)$. 
For $S = N$, the claim follows from Lemma~\ref{lemma:core}. 
Let $S \subset N$. Define the payoff vector $\overline{x} \in \mathbb{R}^S$ by
\[
\overline{x}_j 
= \sum_{\substack{i \in S \\ (i,j)\in E}} a_{ij}
  + \min_{\substack{k \notin S \\ (k,j)\in E}} a_{kj},
\qquad j \in S.
\]

We first verify efficiency. By definition of the game,
\[
v_S(S)
= \sum_{i \in S} \sum_{\substack{j \in S\\ (i,j)\in E}} a_{ij}
  + \sum_{i \in S^*} \min_{\substack{k \notin S \\ (k,i)\in E}} a_{ki},
\]
where
\[
S^* = \{\, i \in S \mid \exists\, k \notin S \text{ such that } (k,i)\in E \,\}.
\]
Rewriting the above expression yields
\[
v_S(S) = \sum_{j \in S} \overline{x}_j,
\]
so $\overline{x}$ is efficient.

    Now let us compare $v_S(T)$ and the $\sum_{j\in T}\overline{x}_j $. At first we establish that if the following inequality holds \begin{equation}\label{eq:min_monotonicity}
        \min_{i\notin S}a_{ij}\geq \min_{i\notin T}a_{ij}, \ T\subseteq S \ \text{if} \ \exists (i,j)\in E: i\notin S.
    \end{equation}
    Now we show coalitional rationality for $(S,v_S)$.  Let $T\subset S$, then
    \begin{equation*}
        \sum_{j\in T}\overline{x}_j= \sum_{i\in S}\sum_{\substack{j \in T\\ (i,j)\in E}}a_{ij}+\sum_{j\in S^*\cap T}\min_{i\notin S}a_{ij}=\sum_{i\in T}\sum_{\substack{j \in T\\ (i,j)\in E}}a_{ij}+\sum_{i\in S\setminus T}\sum_{\substack{j \in S\\ (i,j)\in E}}a_{ij}+\sum_{j\in S^*\cap T}\min_{i\notin S}a_{ij}.
    \end{equation*}
    Since 
    \begin{equation}\label{eq:inequal_tb}
        v(T)=\sum_{i\in T}\sum_{\substack{j \in T\\ (i,j)\in E}}a_{ij}+\sum_{j\in T^*}\min_{i\notin T}a_{ij},
    \end{equation}
    we can reduce the proof to show only the following inequality
    \begin{equation*}
       \sum_{i\in S\setminus T}\sum_{\substack{j \in T\\ (i,j)\in E}}a_{ij}+\sum_{j\in S^*\cap T}\min_{i\notin S}a_{ij}\geq \sum_{j\in T^*}\min_{i\notin T}a_{ij}
    \end{equation*}

The proof proceeds by considering three cases according to the existence of incoming edges to players in $T$.

\begin{enumerate}
    \item $T\cap S^*\neq \emptyset$ (i.e., there exists $j\in T$ with an incoming edge from $N\setminus S$).

    In this case, the second term on the left-hand side is present. For each such $j\in T\cap S^*$, inequality~\eqref{eq:min_monotonicity} implies
    \[
    \min_{\substack{k\notin S\\(k,j)\in E}} a_{kj}
    \;\ge\;
    \min_{\substack{k\notin T\\(k,j)\in E}} a_{kj}.
    \]
    The remaining terms on the left-hand side are nonnegative by $a_{ij}\ge 0$, hence the desired inequality follows.

    \item $T\cap S^*=\emptyset$ but $T^*\neq \emptyset$ (i.e., $T$ has incoming edges from $S\setminus T$, but none from $N\setminus S$).

    Then the second term on the left-hand side vanishes, and the inequality reduces to showing that the cross-term
    \[
    \sum_{i\in S\setminus T}\sum_{\substack{j\in T\\(i,j)\in E}} a_{ij}
    \]
    dominates the external term of $v(T)$:
    \[
    \sum_{j\in T^*}\min_{\substack{k\notin T\\(k,j)\in E}} a_{kj}.
    \]
    For each $j\in T^*$, all incoming neighbors of $j$ outside $T$ lie in $S\setminus T$ by assumption, hence
    \[
    \min_{\substack{k\notin T\\(k,j)\in E}} a_{kj}
    =
    \min_{\substack{k\in S\setminus T\\(k,j)\in E}} a_{kj}
    \le
    \sum_{\substack{i\in S\setminus T\\(i,j)\in E}} a_{ij},
    \]
    where the last inequality uses nonnegativity. Summing over $j\in T^*$ yields the claim.

    \item $T^*=\emptyset$ (i.e., $T$ has no incoming edges from $N\setminus T$).

    In this case, the external term on the right-hand side is zero, and the inequality holds because all terms on the left-hand side are nonnegative.
\end{enumerate}

\end{proof}

\begin{comment}
\subsection{What to do next}
Try to establish vector value. Informativness or equlibria in the 
What are pros and cons of the coop game theory approach. Whtat about the nonmonotonicity of external game, should not be the grand coalition value be a one.
Fit to conference like audience.
\end{comment}

\section{Conclusion and future directions}\label{sec:conclusion}

In this paper we introduced a new class of cooperative games induced by weighted directed graphs. The coalitional value combines two components: an additive internal interaction term and an external exposure term based on minimal incoming evaluations from outside the coalition.

Despite this non-additive structure, the game admits an unanimity decomposition with simple structure. This representation enables the derivation of polynomial-time computable formulas for both the Shapley and the Banzhaf values. In particular, the external component generates a nested chain of unanimity games determined by the ranking of incoming neighbors, which makes the computation of Shapley and Banzhaf values tractable.

From the perspective of cooperative game theory, we established several structural properties. The game has a nonempty core, which in fact reduces to a unique allocation equal to the vector of total incoming weights. Moreover, the class is totally balanced, implying that every subgame admits a nonempty core and is stable against coalitional deviations. Notably, the unique core allocation differs from the Shapley value, which highlights the distinction between stability-based and fairness-based solution concepts in this setting.

The class therefore provides an tractable example of a game in which Shapley value, stability properties, and sensitivity to local perturbations can all be characterized explicitly.

\paragraph{Future Research Directions}

Several natural research directions follow from the present analysis.

A first line of research concerns parametric generalizations of the coalitional value. 
Instead of the sum of the internal and external components, one may consider convex combinations of the components.
It would be of interest to determine how structural properties such as total balancedness or uniqueness of the core  depend on the parameter $\alpha$.

A second direction concerns alternative specifications of the external component. 
Replacing the minimum operator by other aggregation rules, for instance, average, median, or more general order statistics of incoming links may produce different classes of games. 
A systematic investigation of which external operators preserve tractability of power indices and stability properties remains an open problem.

A third direction involves strategic considerations and mechanism design. Since local edge weights determine global allocations, individual players or coalitions may attempt to manipulate reported evaluations. Designing aggregation rules that guarantee incentive compatibility or impartiality, in the sense that no agent can influence their own score, while preserving fairness properties such as the Shapley value represents a concrete and largely open problem.

\section*{Acknowledgements}
This work was supported by the Czech Science Foundation (GAČR) under Grant No.~\texttt{25-17221S}.
% ---------- Bibliography ----------
% If you use BibTeX:
\bibliographystyle{plain} % swap to alpha / apalike / unsrt / etc.
\bibliography{references}

% If you use biblatex instead (often nicer control), remove the above and use:
% \usepackage[backend=biber,style=authoryear]{biblatex}
% \addbibresource{references.bib}
% ...
% \printbibliography

\begin{comment}
\appendix
\section{Additional Proofs}
Optional appendix.
\end{comment}

\end{document}